\documentclass[runningheads]{llncs}
 
\usepackage[T1]{fontenc}

\usepackage{graphicx}

\usepackage{makeidx}   
\usepackage{graphicx}
\usepackage{color}
\usepackage{lineno,hyperref}

\usepackage{amsfonts}
\usepackage{amsmath}
\usepackage{amsthm}  
\usepackage{enumerate}
\usepackage{makecell}
\usepackage{appendix}

\usepackage{mathrsfs}
\usepackage{cases}
\usepackage{xcolor}

\newtheorem{reduction}{Reduction Rule}

\usepackage{tcolorbox}

\usepackage{fancyhdr}
\usepackage{ifthen}
\usepackage{todonotes}

\newcommand{\liu}[1]{{\color{black} #1}}

\begin{document}
 
\title{Linear Kernels for $l$-Exact Component Order Connectivity}

\author{
    Yuxi Liu\orcidID{0009-0009-2171-9042}
    \and
    Mingyu Xiao\orcidID{0000-0002-1012-2373}
}
 
\authorrunning{Y. Liu and M. Xiao}

\institute{University of Electronic Science and Technology of China, Chengdu, China
\email{yuxiliu823@gmail.com, myxiao@uestc.edu.cn}}

\maketitle               
 
\begin{abstract}
    The \textsc{$l$-Exact Component Order Connectivity} problem asks whether, given an input graph $G$ and an integer $k$, there exists a vertex subset $S\subseteq V(G)$ of size at most $k$ such that every connected component in $G - S$ has exactly $l$ vertices. In this paper, we present an $O(kl)$-vertex kernel for this problem, computable in $|V(G)|^{O(l)}$ time. This is the first known linear kernel for each fixed $l\geq 3$. For $l=1$, this problem reduces to the classical \textsc{Vertex Cover}, and our result matches the best-known $2k$-vertex kernel. For $l=2$ (known as \textsc{Deletion to Induced Matching}), we can get a $(3k + 1)$-vertex kernel,  improving the previously known result of $6k$ vertices. Our kernelization algorithm is built upon on an extended crown decomposition combined with linear programming and other techniques.

\end{abstract}

\section{Introduction}
Given a graph $G$ and an integer $k$,
the classic \textsc{Vertex Cover} problem asks whether there exists a vertex subset $S$ of size at most $k$ such that  every connected component in the remaining graph after deleting $S$ contains exactly one vertex.
This structural perspective naturally motivates the following generalization: can we delete at most $k$ vertices from $G$ such that every connected component in the resulting graph contains exactly $l$ vertices?
Formally, for every integer $l \geq 1$, we consider the following problem, called \textsc{$l$-Exact Component Order Connectivity} ($l$-ECOC).

\begin{tcolorbox} 
\textsc{$l$-Exact Component Order Connectivity} ($l$-ECOC)\\
\textbf{Instance:} A graph $G=(V, E)$ and an integer $k$.\\
\textbf{Question:} Does there exist a vertex subset $S\subseteq V$ of size at most $k$ such that every connected component in $G - S$ has exactly $l$ vertices?
\vspace{-2mm}
\end{tcolorbox}

\vspace{2mm}
\noindent
\textbf{Related work.}
When $l = 1$, \textsc{$1$-Exact Component Order Connectivity} corresponds exactly to \textsc{Vertex Cover}, which has been extensively studied~\cite{balasubramanian1998improved,buss1993nondeterminism,chandran2005refined,chen2001vertex,chen2000improvement,niedermeier1999upper,niedermeier2003efficient}.
For a long period time, the algorithm by Chen et al.~\cite{chen2010improved} held the best-known running time of $O^*(1.2738^k)$.
Recently, this result was improved to $O^*(1.25284^k)$ by Harris and Narayanaswamy~\cite{DBLP:conf/stacs/0001N24}.
When $l = 2$, \textsc{$2$-Exact Component Order Connectivity} is known as 
\textsc{Deletion to Induced Matching}.
Xiao and Kou \cite{xiao2020parameterized} showed that \textsc{Deletion to Induced Matching} can be solved in $O^*(1.7485^k)$ time, a bound independently achieved by Kumar et al.~\cite{kumar2020deletion}.
Recently, this running time was improved to $O^*(1.6477^k)$ time by Liu and Xiao~\cite{liu2025improved}.

In terms of kernelization, 
Chen et al.~\cite{chen2001vertex} first observed that the Nemhauser-Trotter theorem~\cite{nemhauser1974properties} directly yields a $2k$-vertex kernel for \textsc{Vertex Cover}.
This is considered a fundamental result in the field of parameterized complexity.
For \textsc{Deletion to Induced Matching}, Moser and Thilikos~\cite{moser2009parameterizedregular} provided a kernel with $O(k^3)$ vertices, which Mathieson and Szeider~\cite{mathieson2012editing} subsequently improved to $O(k^2)$ vertices.
Xiao and Kou~\cite{xiao2020parameterized} obtained the first linear vertex kernel of $7k$ vertices for this problem, which was recently improved to $6k$~\cite{liu2025improved}.

A closely related problem is the \textsc{$l$-Component Order Connectivity} problem, which asks whether there exists a vertex set $S$ of size at most $k$ such that every connected component in $G - S$ has at most $l$ vertices.
Drange et al.~\cite{drange2016computational} provided a parameterized algorithm running in $O^*((l + 1)^k)$ time, along with a kernel of size $O(kl(l+k))$.
The parameterized algorithm can be improved to $O^*((l + 1 - 0.9245)^k)$ time by reducing to \textsc{$(l + 1)$-Hitting Set} and applying the corresponding parameterized algorithm~\cite{fomin2010iterative}.
The kernel size was improved to $9lk$ vertices by Xiao~\cite{xiao2017linear}, and further to $3lk$ vertices by Casel et al.~\cite{casel2020balanced}.
Kumar and Lokshtanov~\cite{DBLP:conf/iwpec/KumarL16} showed that this problem admits a kernel with $2lk$ vertices computable in $|V(G)|^{O(l)}$ time.

\vspace{2mm}
\noindent
\textbf{Our contributions.}
In this paper, we show that \textsc{$l$-Exact Component Order Connectivity} admits a kernel with at most $(l + 1)k + l - 1$ vertices computable in $|V(G)|^{O(l)}$ time.
Specifically, for $l = 1$, our kernel matches the classic $2k$-vertex kernel for \textsc{Vertex Cover}.
For $l = 2$, we obtain a kernel with $3k + 1$ vertices for \textsc{Deletion to Induced Matching}, improving upon the previous best bound of $6k$ vertices.
Furthermore, our result provides the first kernel with a linear number of vertices for any fixed $l \geq 3$.

To achieve this result, we generalize the classic VC crown decomposition for \textsc{Vertex Cover} to \textsc{$l$-Exact Component Order Connectivity},  introducing the notion of an \emph{ECOC crown decomposition}.
Similar to VC crown decomposition,
we demonstrate that whenever an ECOC crown decomposition is given, the current instance can be safely reduced.
Then, through combinatorial analysis, we prove that any sufficiently large graph is guaranteed to contain an ECOC crown decomposition.

Based on these analytical tools, our kernelization algorithm first exhaustively applies a set of basic reduction rules.
If the current instance is sufficiently small, the algorithm returns it and terminates.
Otherwise, the algorithm employs linear programming techniques to compute an ECOC crown decomposition and further reduce the instance.
If no such decomposition can be found, the algorithm correctly returns a trivial no-instance.

\vspace{2mm}
\noindent
\textbf{Organization.}
Section 2 reviews fundamental definitions and the classic VC crown decomposition.
Section 3 formalizes the concept of the ECOC crown decomposition, which serves as the core technical tool of our algorithm.
Section 4 presents the algorithm and its analysis, while the proof for the correctness of the final reduction rule is deferred to Section 5.
Section 6 concludes the paper.

\section{Preliminaries}
All graphs considered in this paper are simple and undirected.
Let $G = (V, E)$ be a graph with $n = |V|$ vertices and $m = |E|$ edges.
A singleton $\{v\}$ may simply be denoted by $v$.
We use $V(G')$ and $E(G')$ to denote the vertex set and edge set of a graph $G'$, respectively.
A vertex $v$ is a \emph{neighbor} of a vertex $u$ if there is an edge $uv \in E$.
Let $N_G(v)$ denote the set of neighbors of $v$ in $G$. 
When the graph $G$ is clear from the context, we simply write $N(v)$.
For a vertex subset $X \subseteq V$, let $N(X) = \bigcup_{v \in X} N(v) \setminus X$ and $N[X] = N(X) \cup X$.
For a vertex subset $X \subseteq V$, the subgraph induced by $X$ is denoted by $G[X]$, and $G[V \setminus X]$ is also written as $G \setminus X$ or $G - X$.
Two disjoint vertex subsets $X_1$ and $X_2$ are \emph{adjacent} if there is an edge $uv \in E$ with $u \in X_1$ and $v \in X_2$.

A \emph{path} $P$ of length $t-1$ is a sequence of distinct vertices $v_1, v_2, \dots, v_t$ such that $v_iv_{i+1} \in E$ for all $1 \leq i < t$. 
A vertex $u$ is \emph{reachable} from a vertex $v$ if there exists a path starting at $u$ and ending at $v$. 
By definition, every vertex is reachable from itself. 
A \emph{connected component} of a graph is a maximal set of vertices such that every pair of vertices in the set is mutually reachable. 
For a graph $G$, we denote the collection of its connected components by $\mathcal{C} = \{C_1, C_2, \dots, C_t\}$.
For a collection of connected components $\mathcal{C}$, we use $V(\mathcal{C})$ to denote $\bigcup_{C \in \mathcal{C}} V(C)$.
We say two instances $I$ and $I'$ are \emph{equivalent} if $I$ is a yes-instance if and only if $I'$ is a yes-instance.

The classic crown decomposition technique, called \emph{VC crown decomposition}, was originally introduced for the \textsc{Vertex Cover} problem \cite{DBLP:conf/alenex/Abu-KhzamCFLSS04,chor2004linear}.
In this paper, we will use a variant of it.
We first give the definition of VC crown decomposition for the ease of reference.

\begin{definition}\label{VC-decomposition}
    A \textit{VC crown decomposition} of a graph $G = (V, E)$ is a partition $(I, J, R)$ of the vertex set $V$ satisfying the following properties.
    \begin{enumerate}
        \item There is no edge between $I$ and $R$.
        \item $I$ is an independent set.
        \item There is an injective mapping (matching) $M: J\rightarrow I$ such that $x$ is adjacent to $M(x)$ for all $x\in J$.
    \end{enumerate}
\end{definition}

The following lemma is used to find a VC crown decomposition in polynomial time if it exists.

\begin{lemma}[\cite{DBLP:conf/iwpec/KumarL16}]\label{q-expansion-finding}
    There exists a polynomial time algorithm that given a bipartite graph $G = (A \cup B, E)$ with bipartition $(A, B)$, 
    outputs (if it exist) two sets $I \subseteq A$ and $J\subseteq B$ with an injective mapping (matching) $M: J\rightarrow I$ such that $(I, J, V(G)\setminus (I\cup J))$ is a VC crown decomposition for $G$.
\end{lemma}

\section{ECOC Crown Decomposition}

First, we present a variant of VC crown decomposition for $l$-ECOC, called \emph{ECOC crown decomposition}.

\begin{definition}\label{lECOC-LOS}
    An \textit{ECOC crown decomposition} of a graph $G = (V, E)$ is a partition $(I, J, R)$ of the vertex set $V$ satisfying the following properties:
    \begin{enumerate}
        \item Let $\mathcal{C}$ be the collection of connected components of $G[I]$. For every $C \in \mathcal{C}$, $|C| = l$.
        \item There is no edge between $I$ and $R$ (i.e., $N(I) \subseteq J$).
        \item There is an injective mapping (matching) $M: J \rightarrow \mathcal{C}$ such that for all $x \in J$, there exists a vertex $u \in V(M(x))$ with $xu \in E$.
    \end{enumerate}
\end{definition}

    \begin{figure}[!t]
        \centering
        \includegraphics[scale=0.5]{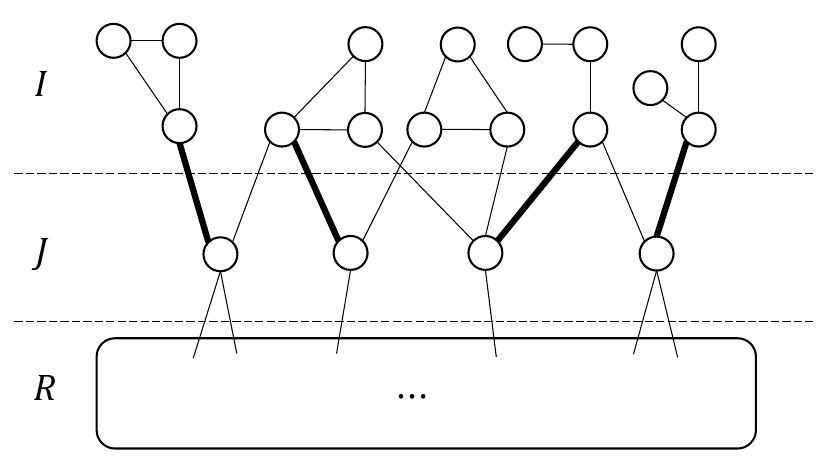}
        \caption{
        An ECOC crown decomposition $(I, J, R)$ with $l = 3$. The four bold edges form an injective matching from $J$ to the collection of connected components $\mathcal{C}$ of $G[I]$.
        }
        \label{Fig:1}
    \end{figure}

See Fig.~\ref*{Fig:1} for an illustration.
When $l = 1$, an ECOC crown decomposition is equivalent to a VC crown decomposition.
When $l = 2$, it is equivalent to the AIM crown decomposition used in~\cite{liu2025improved}.
The following lemma establishes that an ECOC crown decomposition allows us to find parts of the solution.

\begin{lemma}\label{$l$-ECOC-correctness}
    Let $(G, k)$ be an instance of $l$-ECOC and let $(I, J, R)$ be an ECOC crown decomposition of $G$. 
    Then $(G, k)$ is equivalent to the instance $(G \setminus (J\cup I), k - |J|)$.
\end{lemma}

\begin{proof}
    To prove this lemma, it is sufficient to show that there is an optimal solution $S'$ such that $J \subseteq S'$ and $I\cap S' = \emptyset$.
    Let $S$ be an arbitrary optimal solution. 
    We partition $J$ into $ J_1= S\cap J$ and $J_2 = J\setminus J_1$.
    Let $\mathcal{C}$ be the set of the connected components of $G[I]$. 
    Let $\mathcal{C}_2$ be the set of the connected components in $G[I]$ that are adjacent to some vertex in $J_2$. 
    Since there is an injective mapping from $J$ to $\mathcal{C}$ and every component in $\mathcal{C}$ has size $l$, we have that 
    \begin{equation}
        |V(\mathcal{C}_2)|\geq l|J_2|. 
    \end{equation}

    Let $G' = G[V\setminus S]$ be the remaining graph after removing $S$. 
    Let $\mathcal{C}'$ be the set of the connected components in $G'$ that contain at least one vertex in $J_2$.
    Let $T = V(\mathcal{C}')\cup V(\mathcal{C}_2)$. 
    We will focus on the vertex set $T$ in the remaining part of this proof.
    Let $D = V(\mathcal{C}') \setminus V(\mathcal{C}_2)$.
    Note that $J_2\subseteq D$. 
    We construct a new solution candidate $S' = S\setminus V(\mathcal{C}_2) \cup D$. 
    See Fig. \ref*{Fig:2} for an illustration.

    \begin{figure}[!t]
        \centering
        \includegraphics[scale=0.23]{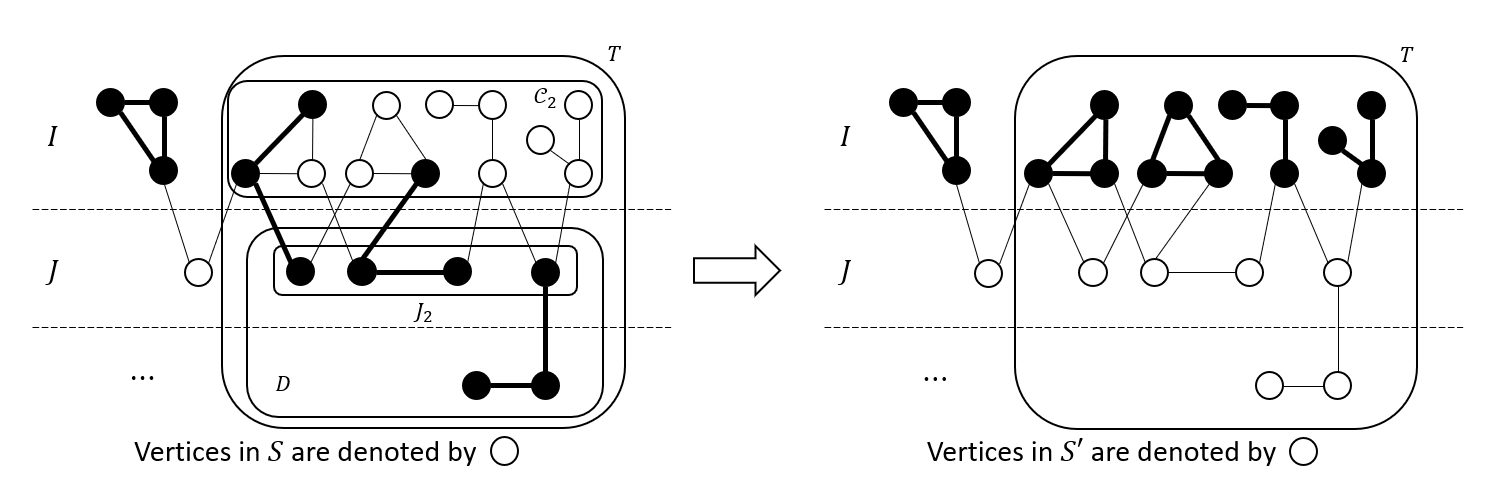}
        \caption{Sets $S$ and $S'$ in the proof of Lemma \ref*{$l$-ECOC-correctness}.}
        \label{Fig:2}
    \end{figure}

    Let $S_1 = S\setminus (J \cup I)$. By definition of $S$ and $D$, each componenent in $G[V\setminus(J\cup I\cup D)] - S_1$ has size $l$.
    Since $(J\cup D)\subseteq S'$ and $I\cap S' = \emptyset$, each componenent in $G[J\cup I\cup D] - (S'\setminus S_1)$ has size $l$.
    Thus, it is easy to see that $S'$ is a feasible solution. 

    Now we show that $S'$ is optimal.
    By the definition of $S'$, we have that
    \begin{equation}
        |S'| = |S| - |S\cap V(\mathcal{C}_2)| + |D|.  
    \end{equation}

    By the definition of $\mathcal{C}'$, we have that $V(\mathcal{C}') \cap S = \emptyset$. So we have that
    \begin{equation}
        |S\cap T| = |S\cap V(\mathcal{C}_2)|.  
    \end{equation}

    Recall that each component in $G[I]$ has size $l$.
    By the definition of $\mathcal{C}_2$, 
    if there is a component $C$ in $G'$ containing one vertex in $V(\mathcal{C}_2)$, the component $C$ must contain one vertex in $J_2$.
    Thus, there are at most $|J_2|$ connected components in $G'[T\setminus S]$.
    We get that
    \begin{equation}
        |S\cap T| \geq |T| - l|J_2|.  
    \end{equation}

    By inequality (1) and the definition of $T$, we have that
    \begin{equation}
        |D| = |T| - |V(\mathcal{C}_2)|\leq |T| - l|J_2|.  
    \end{equation}

    Equality (3) with inequalities (4) and (5) together imply that
    \begin{equation}
        |D| \leq |S\cap V(\mathcal{C}_2)|.  
    \end{equation}

    Equality (2) with inequality (6) together imply that
    \[
        |S'|=|S| - |S\cap V(\mathcal{C}_2)| + |D| \leq |S|.  
    \]

    Since $S$ is optimal, we know that $S'$ is optimal.
    Thus, there is a feasible optimal solution $S'$ such that $J \subseteq S'$ and $I\cap S' = \emptyset$. 
     
\end{proof}

We say that an ECOC crown decomposition $(I, J, R)$ is \emph{strict} if $|\mathcal{C}| \geq |J| + 1$, where $\mathcal{C}$ is the collection of connected components of $G[I]$.
For example, the ECOC crown decomposition illustrated in Fig.~\ref*{Fig:1} is strict since $|\mathcal{C}| = |J| + 1$.
Clearly, every strict ECOC crown decomposition is also an ECOC crown decomposition.
The following lemma demonstrates how to obtain a strict ECOC crown decomposition.

\begin{lemma}\label{strict-LOS-Lemma}   
    Let $G = (V, E)$ be a graph.
    Let $I \subseteq V$ be a non-empty vertex set, let $J = N(I)$, and let $\mathcal{C}$ be the collection of connected components of $G[I]$. 
    If $|C| = l$ for all $C \in \mathcal{C}$ and $|\mathcal{C}| \geq |J| + 1$, then there exists a strict ECOC crown decomposition $(I', J', R')$ such that $\emptyset \neq I' \subseteq I$ and $J' \subseteq J$.
    Furthermore, this strict ECOC crown decomposition can be found in polynomial time if $I$ is given.
\end{lemma}

\begin{proof}
We construct an auxiliary bipartite graph $G' = (V', E')$ with $V' = (A', B')$ as follows: 
    \begin{itemize}
        \item Each vertex $v \in A'$ corresponds to a connected component of size $l$ of $G[I]$.
        \item Each vertex $v \in B'$ corresponds to a vertex in the vertex set $J$. 
        \item An edge $(a, b)$ exists in $E'$ if and only if the component in $G[I]$ corresponding to $a$ is adjacent to the vertex corresponding to $b$ in $G$. 
    \end{itemize}
     
    Since $|\mathcal{C}| \geq |J| + 1$, we have that $|A'|\geq |B'| + 1$. 
    We compute a maximum matching $M'$ in $G'$, which can be done in $O(n^{5/2})$ time using the Hopcroft-Karp algorithm~\cite{hopcroft1973n}.
    The following arguments are based on the maximum matching $M'$.
    A vertex is called $M'$-saturated if it is an endpoint of an edge in $M'$. 
    A path in $G'$ that alternates between edges not in $M'$ and edges in $M'$ is called an $M'$-alternating path.

    Let $Z \subseteq A'$ be the set of vertices in $A'$ that are not $M'$-saturated.
    Since $|A'| \geq |B'| + 1$, 
    the matching $M'$ cannot saturate all vertices in $A'$.
    Thus, we have that $Z \neq \emptyset$.
    We define $I'' \subseteq A'$ and $J'' \subseteq B'$ as the sets of vertices reachable from $Z$ via $M'$-alternating paths in $G'$. 
    By definition, $Z \subseteq I''$. 
    Let $\mathcal{C}'$ and $J'$ be the components and vertices in $G$ corresponding to $I''$ and $J''$, respectively. 
    Let $I' = V(\mathcal{C}')$.
    We verify that $(I', J', V\setminus(I'\cup J'))$ forms the desired strict crown decomposition:

    \begin{itemize}
        \item \textbf{Property 1:} Since every vertex in $A'$ corresponds to a component of size $l$ in $G[I]$, the connected components of $G[I']$ are exactly $\mathcal{C}'$, and each is of size $l$.
        \item \textbf{Property 2:} We show $N(I')\subseteq J'$.
        It suffices to show that in $G'$, the neighborhood of $I''$ is a subset of $J''$. Consider any edge $(a, b) \in E'$ with $a \in I''$.
        Since $a$ is reachable from $Z$ via an alternating path, we analyze the path ending at $a$.
        If $(a, b) \in M'$, then $b$ must have been the predecessor of $a$ on the alternating path, meaning $b \in J''$ is already established.
        If $(a, b) \notin M'$, we extend an alternating path ending at $a$ by the edge $(a, b)$, placing $b$ in $J''$. 
        Thus, it holds that $N(I')\subseteq J'$.
        \item \textbf{Property 3:} Let $M'' = \{(a, b) \in M' \mid b \in J''\}$. 
        Every vertex $b \in J'$ must be saturated by $M'$; otherwise, an alternating path from $Z$ to an unsaturated vertex $b$ would constitute an $M'$-augmenting path, which contradicts the maximality of $M'$~\cite{berge1957two}. 
        For any edge $(a, b) \in M''$ with $b \in J''$, the vertex $a$ must belong to $I''$ because $a$ is trivially reachable by extending the alternating path from $b$ via the matched edge $(a, b)$. This implies $M''$ perfectly matches $J''$ into $I''$, yielding the required injective mapping.
    \end{itemize}

    Note that $I''$ contained exactly $|J''|$ vertices and at least one vertex in $Z$ since $Z \neq \emptyset$. we have $|A''| \geq |J''| + 1$. 
    Since $I''$ corresponds to the connected components $\mathcal{C}'$ of $G[I']$,
    it holds that $|\mathcal{C}'|\geq |J'| + 1$. 
     
\end{proof}

\section{The Algorithm}
Our kernelization algorithm contains several reduction rules. 
When a reduction rule is applied, we assume that no preceding reduction rule is applicable to the current instance.
We first introduce the following simple reduction rules.

\begin{reduction}\label{bound-reduction}
    If $|V(G)| \leq (l + 1)k + l - 1$, return the current instance and exit the algorithm.
\end{reduction}

\begin{reduction}\label{basic-reduction-1}
    If there exists a connected component $C$ such that $|C| < l$, then return $(G - C, k - |C|)$.
\end{reduction}

\begin{reduction}\label{basic-reduction-2}
    If there exists a connected component $C$ such that $|C| = l$, then return $(G - C, k)$.
\end{reduction}

The correctness of Reduction Rules~\ref*{bound-reduction}--\ref*{basic-reduction-2} is straightforward.
The following lemma is crucial for establishing our last reduction rule.

\begin{lemma}\label{RR3-lemma}
There exists an algorithm that, given an instance $(G, k)$ of $l$-ECOC with at least $(l + 1)k + l $ vertices to which Reduction Rules~\ref*{bound-reduction}--\ref*{basic-reduction-2} are not applicable, runs in $|V(G)|^{O(l)}$ time and either finds an ECOC crown decomposition $(I, J, R)$ of $G$ or concludes that $(G, k)$ is a no-instance.
\end{lemma}

We postpone the proof of Lemma~\ref{RR3-lemma} and the description of the corresponding $|V(G)|^{O(l)}$ time algorithm to the next section.
Based on Lemma~\ref{RR3-lemma}, we obtain the following reduction rule.

\begin{reduction}\label{lp-reduction}
    Call the algorithm mentioned in Lemma~\ref*{RR3-lemma} to compute an ECOC crown decomposition $(I, J, R)$.
    Return the equivalent instance $(G\setminus (J \cup I) , k - |J|)$.
    If the algorithm cannot find such an ECOC crown decomposition, return no to indicate that there is no solution.
\end{reduction}

Now, we analyze the kernel size and the running time of the algorithm.
For Reduction Rules~\ref*{basic-reduction-1}--\ref*{lp-reduction}, the kernelization algorithm either returns a strictly smaller equivalent instance or correctly concludes that no solution exists.
When Reduction Rule~\ref*{bound-reduction} is applied, the algorithm terminates and returns an instance $(G, k)$ with $|V(G)| < (l + 1)k + l$.
Thus, the resulting kernel has at most $(l + 1)k + l - 1$ vertices.

Regarding the running time, Reduction Rules~\ref*{bound-reduction}--\ref*{basic-reduction-2} clearly run in polynomial time.
Reduction Rule~\ref*{lp-reduction} runs in $|V(G)|^{O(l)}$ time by Lemma~\ref*{RR3-lemma}.
Furthermore, each application of Reduction Rules~\ref*{basic-reduction-1}--\ref*{lp-reduction} decreases the number of vertices in the instance by at least one.
Thus, these reduction rules are applied at most $|V(G)|$ times in total.
Consequently, the kernelization algorithm runs in $|V(G)|^{O(l)}$ time.
We have the following theorem.

\begin{theorem}
    \textsc{$l$-Exact Component Order Connectivity} admits a kernel with $(l + 1)k + l - 1$ vertices that can be computed in $|V(G)|^{O(l)}$ time.
\end{theorem}

When $l = 2$, this immediately yields the following theorem for the \textsc{Deletion to Induced Matching} problem.

\begin{theorem}
    \textsc{Deletion to Induced Matching} admits a kernel with $3k + 1$ vertices.
\end{theorem}

\section{The proof of Lemma~\ref*{RR3-lemma}} 
In this section, we prove Lemma~\ref*{RR3-lemma} and introduce the algorithm used in Reduction Rule~\ref*{lp-reduction}.

Firstly, we present an (incomplete) Integer Programming model for $l$-ECOC.  
We introduce $n = |V(G)|$ integer variables $\{x_v\}_{v\in V(G)}$, where $x_v = 1$ indicates that $v$ is selected in the solution $S$, and $x_v = 0$ indicates that it is not. 
Clearly, we have that $|S| = \sum_{v\in V} x_v$.
To ensure $C\cap S \neq \emptyset$ for each connected set $C$ of size $l + 1$, we introduce the constraint $\sum_{v\in C} x_v\geq 1$.
This gives us the following IP model, denoted as WECOC-IP:
\[
    \begin{split}
        \mbox{Min } &\sum_{v\in V} x_v\\  
           \mbox{Subject to } & \sum_{v\in C} x_v\geq 1, \mbox{  for every connected set C of size } l + 1,\\     
        & x_v\in \{0, 1\},\quad \forall v\in V.
    \end{split}
\]
Note that a solution to WECOC-IP does not necessarily correspond to a feasible solution for $l$-ECOC,
as the connected components in the remaining graph may have a size less than $l$.
The LP relaxation of this model, denoted by WECOC-LP, is obtained by replacing the constraint $x_v\in \{0, 1\}$ with the constraint $0\leq x_v \leq 1$.
Let $S_L$ be an optimal solution to WECOC-LP with objective value $L = \sum_{v \in V(G)} x_v$. 
For any set of vertices $X \subseteq V$, we use the notation $X = 1$ (resp. $X = 0$) to signify that $x_v = 1$ (resp. $x_v = 0$) for all $v \in X$.

The following lemma establishes that a strict ECOC crown decomposition is guaranteed to exist whenever the vertex set is sufficiently large. 
This size condition directly determines the upper bound on the kernel size.

\begin{lemma}\label{strict-key-bound}
Let $(G, k)$ be an instance where Reduction Rules 1--\ref*{basic-reduction-2} are not applicable.
If $(G, k)$ is a yes-instance, then there exists a strict ECOC crown decomposition in $G$.
\end{lemma}

\begin{proof}
    Consider any minimum solution $S$, let $A = V(G)\setminus S$.
    Let $\mathcal{C}$ be the set of the connected components of $G[A]$.
    Since $S$ is a feasible solution, we have that $\forall C\in \mathcal{C}, |C| = l$.
    Since Reduction Rule 1 is not applicable, we have that $|V(G)| \geq (l + 1)k + l$.
    Since $(G, k)$ is a yes-instance, we have that $|S|\leq k, |A| = |V(G)| - |S|\geq lk + l$ and then $|\mathcal{C}|\geq k + 1\geq |S| + 1$.  
    Hence, by Lemma \ref*{strict-LOS-Lemma}, we can conclude that $G$ has a strict ECOC crown decomposition $(I', J', R')$ with $\emptyset \neq I'\subseteq A$.
\end{proof}

We call a strict ECOC crown decomposition $(I,J,R)$ \emph{minimal} if there do not exist proper subsets $I' \subset I$ and $J' \subset J$ such that $(I',J',V\setminus (I'\cup J'))$ is a strict ECOC crown decomposition.

\begin{lemma}\label{LP-lemma}
    Let $(G, k)$ be a yes-instance where Reduction Rules 1--\ref*{basic-reduction-2} are not applicable.
    For any minimal strict ECOC crown decomposition $(I, J, R)$, an optimal LP solution $S_L$ to WECOC-LP satisfies $I = 0$ and $J = 1$.
\end{lemma}

\begin{proof}
    Let $\mathcal{C}$ be the set of connected components of $G[I]$. We partition $\mathcal{C}$ into two sets: let $\mathcal{C}_1$ be the set of components where every vertex is assigned a value of $0$ by $S_L$, and let $\mathcal{C}_2 = \mathcal{C} \setminus \mathcal{C}_1$ be the remaining components (i.e., each component in $\mathcal{C}_2$ contains at least one vertex $u$ with $x_u > 0$). We also partition $J$ into $J_1 = N(V(\mathcal{C}_1))$ and $J_2 = J \setminus J_1$.

For any component $C \in \mathcal{C}_1$ and any adjacent vertex $v \in J_1$, the set $V(C) \cup \{v\}$ forms a connected set of size $l+1$. 
    The LP constraints require $\sum_{u \in V(C)} x_u + x_v \geq 1$. 
    Since $x_u = 0$ for all $u \in V(C)$ by the definition of $\mathcal{C}_1$, it must be that $x_v = 1$. 
    Thus, $x_v = 1$ holds for all $v \in J_1$. 
    We will show that $\mathcal{C}_2 = \emptyset$ to complete this proof. 
    Assume for the sake of contradiction that $\mathcal{C}_2 \neq \emptyset$. 
    We consider two cases based on the cardinality of $\mathcal{C}_1$.

    If $|\mathcal{C}_1| \geq |J_1| + 1$, Lemma~\ref*{strict-LOS-Lemma} guarantees the existence of a strict ECOC crown decomposition $(I', J', R')$ with $I' \subseteq V(\mathcal{C}_1)$ and $J' \subseteq J_1$. 
    Since $\mathcal{C}_2 \neq \emptyset$, it holds that $I'\subset I$, which directly contradicts the assumption that $(I, J, R)$ is a minimal strict ECOC crown decomposition.

    If $|\mathcal{C}_1| \leq |J_1|$, then it holds that $|\mathcal{C}_2| \geq |J_2| + 1$ since $|\mathcal{C}| \geq |J| + 1$.
    Let $J'_2 = \{v \in J \mid x_v < 1\}$. 
    Since $x_v = 1$ for all $v \in J_1$, we have $J'_2 \subseteq J_2$, which yields $|\mathcal{C}_2| \geq |J'_2| + 1$. 
    Recall that there exists an injective mapping $M: J \rightarrow \mathcal{C}$ \liu{such that every $v \in J$ is adjacent to some vertex in $M(v)$}. 
    For any vertex $v \in J'_2$, the set $V(M(v)) \cup \{v\}$ forms a connected set of size $l+1$, so $\sum_{u \in V(M(v))} x_u + x_v \geq 1$. 
    Since $x_v < 1$, we deduce that $\sum_{u \in V(M(v))} x_u \geq 1 - x_v > 0$. 
    Therefore, $M(v)$ must contain at least one vertex with a positive value, meaning $M(v) \in \mathcal{C}_2$. 
    Let $\mathcal{C}' = \{M(v) \mid v \in J'_2\}$. 
    It is clear that $\mathcal{C}' \subseteq \mathcal{C}_2$. 
    Since $|\mathcal{C}'| = |J'_2|$ and $|\mathcal{C}_2| \geq |J'_2| + 1$, it holds that $|\mathcal{C}'|<|\mathcal{C}_2|$, which implies that $\mathcal{C}' \subset \mathcal{C}_2$.
    \liu{See Fig. \ref*{Fig:3} for an illustration.}

    \begin{figure}[!t]
        \centering
        \includegraphics[scale=0.5]{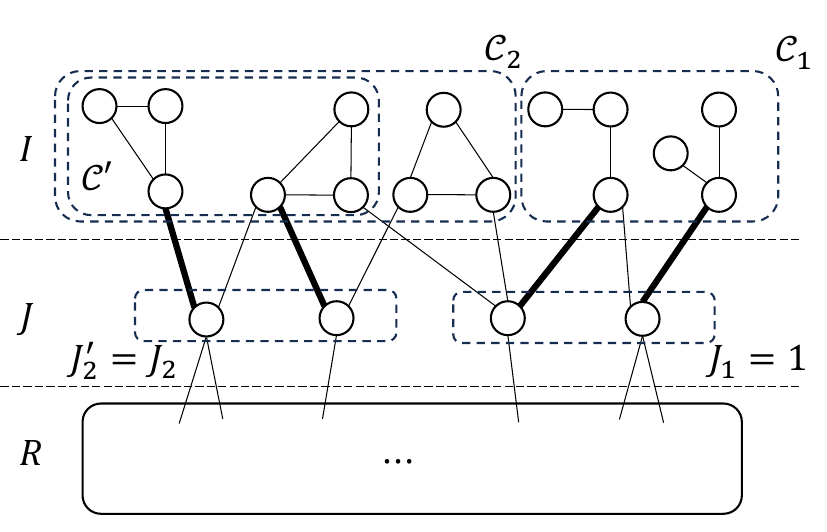}
        \caption{A minimal strict ECOC crown decomposition $(I, J, R)$ where $l = 3$ and $J_2' = J_2$. The four bold edges form an injective matching from $J$ to the collection of connected components $\mathcal{C}$ of $G[I]$.
        }
        \label{Fig:3}
    \end{figure}

    We construct a new LP solution $S_{L'}$ as follows. 
    For each vertex $v \in J_2$, we set $x_v = 1$ (increasing its value by $1 - x_v$). 
    For each vertex $u \in V(\mathcal{C}_2)$, we set $x_u = 0$ (decreasing its value by $x_u$). 
    All other variables remain unchanged. 
     
    \liu{We now verify the feasibility of $S_{L'}$. Consider any connected set $K \subseteq V(G)$ of size $l+1$.
    We consider the following two cases.
If $K \cap I \neq \emptyset$, then $K$ cannot be fully contained within $I$ because every connected component in $G[I]$ has a size of exactly $l$. Thus, $K$ must contain at least one vertex from $N(I)$.
Since $N(I) \subseteq J$, $K$ contains some vertex $v \in J$. By our construction and the fact that $J_1$ was already set to $1$, $S_{L'}$ satisfies $x_v = 1$ for all $v \in J$. Therefore, $\sum_{w \in K} x_w \geq x_v = 1$. If $K \cap I = \emptyset$, then $K \subseteq V(G) \setminus I$. For any vertex $w \in V(G) \setminus I$, its value in $S_{L'}$ is either increased to $1$ (if $w \in J_2$) or left unchanged. Because the original solution $S_L$ is feasible, we have $\sum_{w \in K} x_w \geq 1$. Hence, $S_{L'}$ is a feasible LP solution.}

    Finally, we evaluate the objective value $L'$ of the new solution. The construction increases the sum by $\sum_{v \in J'_2} (1 - x_v)$ \liu{because $1 - x_v = 0$ for all $v \in J_2 \setminus J'_2$. Concurrently}, it decreases the sum by $\sum_{u \in V(\mathcal{C}_2)} x_u$. 
    \liu{Based on the LP constraints for the sets $V(M(v)) \cup \{v\}$, we have established that $\sum_{u \in V(M(v))} x_u \geq 1 - x_v$. Summing this over all $v \in J'_2$ yields:}
    $$ \sum_{u \in V(\mathcal{C}')} x_u = \sum_{v \in J'_2} \sum_{u \in V(M(v))} x_u \geq \sum_{v \in J'_2} (1 - x_v). $$
    \liu{As shown earlier, $\mathcal{C}' \subset \mathcal{C}_2$, meaning the set of components $\mathcal{C}_2 \setminus \mathcal{C}'$ is non-empty. By the definition of $\mathcal{C}_2$, every component in $\mathcal{C}_2 \setminus \mathcal{C}'$ contains at least one vertex with a positive value in $S_L$. This implies that $\sum_{u \in V(\mathcal{C}_2 \setminus \mathcal{C}')} x_u > 0$. Consequently, we obtain the strict inequality:
    $$ \sum_{u \in V(\mathcal{C}_2)} x_u = \sum_{u \in V(\mathcal{C}')} x_u + \sum_{u \in V(\mathcal{C}_2 \setminus \mathcal{C}')} x_u > \sum_{v \in J'_2} (1 - x_v). $$}
     
    This strict inequality implies that $L' < L$, contradicting the assumption that $S_L$ is an optimal LP solution.

    Therefore, it holds that $\mathcal{C}_2 = \emptyset$, meaning $x_u = 0$ for all $u \in I$. 
    Consequently, since the mapping $M$ matches every $v \in J$ to a component in $I$ where all vertices are $0$, the LP constraint requires $x_v = 1$ for all $v \in J$. 
    This concludes the proof.
\end{proof}

Given an optimal solution $S_L$ to WECOC-LP, define $A := \{v \in V \mid x_v = 0\}$ and $B := \{v \in V \mid x_v = 1\}$.
By Lemma~\ref*{LP-lemma}, there exists an ECOC crown decomposition $(I, J, R)$ such that $I \subseteq A$ and $J \subseteq B$.
We are now ready to prove Lemma~\ref*{RR3-lemma}.

\begin{proof}[Proof of Lemma~\ref*{RR3-lemma}]
The algorithm begins by solving  WECOC-LP in $|V(G)|^{O(l)}$ time. 
    Let $A = \{v \in V \mid x_v = 0\}$ and $B = \{v \in V \mid x_v = 1\}$. 
     
    We construct an auxiliary bipartite graph $G' = (V', E')$ with $V' = (A', B')$ as follows: 
    \begin{itemize}
        \item Each vertex $v \in A'$ corresponds to a connected component of size $l$ of $G[A]$.
        \item Each vertex $v \in B'$ corresponds to a vertex in the vertex set $B$. 
        \item An edge $(a, b)$ exists in $E'$ if and only if the component in $G[A]$ corresponding to $a$ is adjacent to the vertex corresponding to $b$ in $G$. 
    \end{itemize}

    Next, we apply the polynomial time algorithm from Lemma \ref*{q-expansion-finding} to $G'$. 
    If it fails to find a VC crown decomposition, we conclude that $(G, k)$ is a no-instance. 
    If it returns a VC crown decomposition $(I', J', R')$ for $G'$, we construct an ECOC crown decomposition $(I^*, J^*, R^*)$ for $G$ such that $I^*$ is the vertices in the components corresponding to the vertices in $I'$ and $J^*$ is the vertices corresponding to the vertices in $J'$. 
    This algorithm clearly runs in $|V(G)|^{O(l)}$ time. We now prove its correctness.

    Let $\mathcal{C}^*$ be the set of the connected components of $G[I^*]$.
    Let us check the three conditions in the definition of ECOC crown decomposition in $(I^*, J^*, R^*)$.
    \begin{enumerate}
        \item Since $I'$ is an independent set, it holds that $\forall C^*\in \mathcal{C}^*, |C^*| = l$.
        \item Since there is no edge between $C'$ and $R'$, we have that $N(I^*)\subseteq J^*$.

        \item Since there is an injective mapping (matching) $M': J'\rightarrow I'$ such that, $\forall x\in J'$, $x$ is adjacent to $M'(x)$, 
        there is an injective mapping (matching from vertices to components) $M: J^*\rightarrow \mathcal{C}^*$ such that $\exists u\in M(x), xu\in E$ holds for all $x\in J^*$ by transforming the vertices in $M'$ into the corresponding vertices or components in $G$.
    \end{enumerate}

    Since $I'\neq \emptyset$, we can say $I^* \neq \emptyset$.
    So $(I^*, J^*, R^*)$ is an ECOC crown decomposition of $V$ with $I^* \neq \emptyset$.

    It remains to show that if $(G, k)$ is a yes-instance, the algorithm will not output a no-instance. 
    Assume $(G, k)$ is a yes-instance. 
    By Lemma \ref{strict-key-bound}, there exists a strict ECOC crown decomposition. 
    Let $(I, J, R)$ be a minimal strict ECOC crown decomposition. 
    By Lemma \ref{LP-lemma}, the optimal LP solution sets $0$ to all vertices in $I$ and $1$ to all vertices in $J$. 
    Therefore, $I \subseteq A$ and $J \subseteq B$. 

    Furthermore, since $N_G(I) \subseteq J \subseteq B$, the components of $G[I]$ are completely isolated from any other vertices in $A$. 
    Let the corresponding vertex set of this component set be $I' \subseteq A'$, and let $J$ correspond to $J' \subseteq B'$ in $G'$. 
    We verify that $(I', J', V(G') \setminus (I' \cup J'))$ is a valid VC crown decomposition in $G'$. 
    The set $I'$ is an independent set because $A'$ is one side of a bipartite graph. 
    The condition $N_G(I) \subseteq J$ directly implies $N_{G'}(I') \subseteq J'$. 
    Finally, the injective mapping from $J$ to the components of $G[I]$ translates directly to a matching from $J'$ to $I'$ in $G'$. 
    Thus, a valid VC crown decomposition exists in $G'$, which implies that  Lemma \ref{q-expansion-finding} will find one. 
    This completes the proof.
\end{proof}

\section{Conclusion}
In this paper, we present the first linear vertex kernel with $O(lk)$ vertices for \textsc{$l$-Exact Component Order Connectivity} for any fixed $l\geq 1$. When $l=2$ (known as \textsc{Deletion to Induced Matching}), we improve the previous result from $6k$ to $3k + 1$.
When $l=1$ (known as the classical \textsc{Vertex Cover}), our result matches the well-known result of $2k$,
which implies our result is relatively tight to some extent. 

Note that our algorithm runs in $|V(G)|^{O(l)}$ time, which is no polynomial in $l$. Whether we can obtain an $O(kl)$-vertex kernel for \textsc{$l$-Exact Component Order Connectivity} using time also polynomial with $l$ is left as an open problem for further study.
For the highly related problem \textsc{$l$-Component Order Connectivity}, there are known 
 $O(kl)$-vertex kernels computable in polynomial time~\cite{xiao2017linear,casel2020balanced}.
 However, it appears that their techniques could not be directly adapted to \textsc{$l$-Exact Component Order Connectivity}.
Additionally, establishing non-trivial lower bounds on kernels for this problem 
is also an interesting topic for further study.

\bibliographystyle{splncs04}
\bibliography{lECOC}

@article{xiao2020parameterized,
  title={Parameterized algorithms and kernels for almost induced matching},
  author={Xiao, Mingyu and Kou, Shaowei},
  journal={Theoretical Computer Science},
  volume={846},
  pages={103--113},
  year={2020},
  publisher={Elsevier}
}

@inproceedings{chor2004linear,
  title={Linear kernels in linear time, or how to save k colors in $O(n^2)$ steps},
  author={Chor, Benny and Fellows, Mike and Juedes, David},
  booktitle={International workshop on graph-theoretic concepts in computer science},
  pages={257--269},
  year={2004},
  organization={Springer}
}

@article{mathieson2012editing,
  title={Editing graphs to satisfy degree constraints: A parameterized approach},
  author={Mathieson, Luke and Szeider, Stefan},
  journal={Journal of Computer and System Sciences},
  volume={78},
  number={1},
  pages={179--191},
  year={2012},
  publisher={Elsevier}
}

@inproceedings{DBLP:conf/iwpec/KumarL16,
  author       = {Mithilesh Kumar and
                  Daniel Lokshtanov},
  editor       = {Jiong Guo and
                  Danny Hermelin},
  title        = {A 2lk Kernel for l-Component Order Connectivity},
  booktitle    = {11th International Symposium on Parameterized and Exact Computation,
                  {IPEC} 2016, August 24-26, 2016, Aarhus, Denmark},
  series       = {LIPIcs},
  volume       = {63},
  pages        = {20:1--20:14},
  publisher    = {Schloss Dagstuhl - Leibniz-Zentrum f{\"{u}}r Informatik},
  year         = {2016},
  bibsource    = {dblp computer science bibliography, https://dblp.org}
}

@inproceedings{DBLP:conf/alenex/Abu-KhzamCFLSS04,
  author       = {Faisal N. Abu{-}Khzam and
                  Rebecca L. Collins and
                  Michael R. Fellows and
                  Michael A. Langston and
                  W. Henry Suters and
                  Christopher T. Symons},
  editor       = {Lars Arge and
                  Giuseppe F. Italiano and
                  Robert Sedgewick},
  title        = {Kernelization Algorithms for the Vertex Cover Problem: Theory and
                  Experiments},
  booktitle    = {Proceedings of the Sixth Workshop on Algorithm Engineering and Experiments
                  and the First Workshop on Analytic Algorithmics and Combinatorics,
                  New Orleans, LA, USA, January 10, 2004},
  pages        = {62--69},
  publisher    = {{SIAM}},
  year         = {2004},
  bibsource    = {dblp computer science bibliography, https://dblp.org}
}

@article{buss1993nondeterminism,
  title={Nondeterminism within p\^{}},
  author={Buss, Jonathan F and Goldsmith, Judy},
  journal={SIAM Journal on Computing},
  volume={22},
  number={3},
  pages={560--572},
  year={1993},
  publisher={SIAM}
}

@article{hopcroft1973n,
  title={An $n^{5/2}$ algorithm for maximum matchings in bipartite graphs},
  author={Hopcroft, John E and Karp, Richard M},
  journal={SIAM Journal on computing},
  volume={2},
  number={4},
  pages={225--231},
  year={1973},
  publisher={SIAM}
}

@article{berge1957two,
  title={Two theorems in graph theory},
  author={Berge, Claude},
  journal={Proceedings of the National Academy of Sciences},
  volume={43},
  number={9},
  pages={842--844},
  year={1957},
  publisher={National Acad Sciences}
}

@article{chen2010improved,
  title={Improved upper bounds for vertex cover},
  author={Chen, Jianer and Kanj, Iyad A and Xia, Ge},
  journal={Theoretical Computer Science},
  volume={411},
  number={40-42},
  pages={3736--3756},
  year={2010},
  publisher={Elsevier}
}

@inproceedings{DBLP:conf/stacs/0001N24,
  author       = {David G. Harris and
                  N. S. Narayanaswamy},
  title        = {A Faster Algorithm for Vertex Cover Parameterized by Solution Size},
  booktitle    = {41st International Symposium on Theoretical Aspects of Computer Science,
                  {STACS} 2024, March 12-14, 2024, Clermont-Ferrand, France},
  year         = {2024},
}

@article{kumar2020deletion,
  title={Deletion to Induced Matching},
  author={Kumar, Akash and Kumar, Mithilesh},
  journal={arXiv:2008.09660},
  year={2020}
}

@article{moser2009parameterizedregular,
  title={Parameterized complexity of finding regular induced subgraphs},
  author={Moser, Hannes and Thilikos, Dimitrios M},
  journal={Journal of Discrete Algorithms},
  volume={7},
  number={2},
  pages={181--190},
  year={2009},
  publisher={Elsevier}
}

@article{liu2025improved,
  title={An improved kernel and parameterized algorithm for deletion to induced matching},
  author={Liu, Yuxi and Xiao, Mingyu},
  journal={Theoretical Computer Science},
  volume={1041},
  pages={115215},
  year={2025},
  publisher={Elsevier}
}

@article{chen2001vertex,
  title={Vertex cover: further observations and further improvements},
  author={Chen, Jianer and Kanj, Iyad A and Jia, Weijia},
  journal={Journal of Algorithms},
  volume={41},
  number={2},
  pages={280--301},
  year={2001},
  publisher={Elsevier}
}

@article{nemhauser1974properties,
  title={Properties of vertex packing and independence system polyhedra},
  author={Nemhauser, George L and Trotter Jr, Leslie E},
  journal={Mathematical programming},
  volume={6},
  number={1},
  pages={48--61},
  year={1974},
  publisher={Springer}
}

@article{drange2016computational,
  title={On the computational complexity of vertex integrity and component order connectivity},
  author={Drange, P{\aa}l Gr{\o}n{\aa}s and Dregi, Markus and van’t Hof, Pim},
  journal={Algorithmica},
  volume={76},
  number={4},
  pages={1181--1202},
  year={2016},
  publisher={Springer}
}

@article{xiao2017linear,
  title={Linear kernels for separating a graph into components of bounded size},
  author={Xiao, Mingyu},
  journal={Journal of Computer and System Sciences},
  volume={88},
  pages={260--270},
  year={2017},
  publisher={Elsevier}
}

@article{casel2020balanced,
  title={Balanced crown decomposition for connectivity constraints},
  author={Casel, Katrin and Friedrich, Tobias and Issac, Davis and Niklanovits, Aikaterini and Zeif, Ziena},
  journal={arXiv preprint arXiv:2011.04528},
  year={2020}
}

@article{fomin2010iterative,
  title={Iterative compression and exact algorithms},
  author={Fomin, Fedor V and Gaspers, Serge and Kratsch, Dieter and Liedloff, Mathieu and Saurabh, Saket},
  journal={Theoretical Computer Science},
  volume={411},
  number={7-9},
  pages={1045--1053},
  year={2010},
  publisher={Elsevier}
}

@article{balasubramanian1998improved,
  title={An improved fixed-parameter algorithm for vertex cover},
  author={Balasubramanian, R and Fellows, Michael R and Raman, Venkatesh},
  journal={Information Processing Letters},
  volume={65},
  number={3},
  pages={163--168},
  year={1998},
  publisher={Elsevier}
}

@article{chandran2005refined,
  title={Refined memorization for vertex cover},
  author={Chandran, L Sunil and Grandoni, Fabrizio},
  journal={Information Processing Letters},
  volume={93},
  number={3},
  pages={125--131},
  year={2005},
  publisher={Elsevier}
}

@article{chen2000improvement,
  title={Improvement on vertex cover for low-degree graphs},
  author={Chen, Jianer and Liu, Lihua and Jia, Weijia},
  journal={Networks: An International Journal},
  volume={35},
  number={4},
  pages={253--259},
  year={2000},
  publisher={Wiley Online Library}
}

@inproceedings{niedermeier1999upper,
  title={Upper bounds for vertex cover further improved},
  author={Niedermeier, Rolf and Rossmanith, Peter},
  booktitle={Annual Symposium on Theoretical Aspects of Computer Science},
  pages={561--570},
  year={1999},
  organization={Springer}
}

@article{niedermeier2003efficient,
  title={An efficient fixed-parameter algorithm for 3-Hitting Set},
  author={Niedermeier, Rolf and Rossmanith, Peter},
  journal={Journal of Discrete Algorithms},
  volume={1},
  number={1},
  pages={89--102},
  year={2003},
  publisher={Elsevier}
}

\end{document}